\documentclass[aps,twocolumn,showpacs,floatfix,superscriptaddress,prl]{revtex4}
\usepackage{amsmath,amsfonts,amssymb,amsthm,graphicx,color,times,bbm,psfrag,dsfont, braket}
\usepackage{enumerate}
\usepackage{color}
\usepackage[T1]{fontenc}
\usepackage[latin9]{inputenc}

\newcommand{\be}{\begin{equation}}
\newcommand{\ee}{\end{equation}}

\newtheorem{theorem}{Theorem}

\newtheorem{example}{Example}

\newcommand{\ZZ}{\ensuremath{\mathbb{Z}}}

\newcommand{\AME}[1]{\ensuremath{\textrm{AME}(#1)}}
\newcommand{\QSS}[1]{\ensuremath{((#1))} threshold QSS scheme}
\newcommand{\floor}[1]{\ensuremath{\lfloor #1 \rfloor}}

\begin{document}

\title{Absolute Maximal Entanglement and Quantum Secret Sharing}

\author{Wolfram Helwig}
\author{Wei Cui}
\affiliation{Center for Quantum Information and Quantum Control (CQIQC),
Department of Physics and Department of Electrical \& Computer Engineering,
University of Toronto, Toronto, Ontario, M5S 3G4, Canada}

\author{Jos\'e Ignacio Latorre}
\affiliation{
Dept. d'Estructura i Constituents de la Mat\`eria,
Universitat de Barcelona, 647 Diagonal, 08028 Barcelona, Spain
}

\author{Arnau Riera}
\affiliation{
Max Planck Institute for Gravitational Physics,
Albert Einstein Institute,
Am M\"uhlenberg 1, D-14476 Golm, Germany
}
\affiliation{
Dahlem Center for Complex Quantum Systems,
Freie Universit\"at Berlin, 14195 Berlin, Germany
}

\author{Hoi-Kwong Lo}
\affiliation{Center for Quantum Information and Quantum Control (CQIQC),
Department of Physics and Department of Electrical \& Computer Engineering,
University of Toronto, Toronto, Ontario, M5S 3G4, Canada}

\date{\today}

\begin{abstract}         
We study the existence of absolutely maximally entangled (AME) states in quantum mechanics and its applications to quantum information. AME states are characterized by being maximally entangled for all bipartitions of the system and exhibit genuine multipartite entanglement. With such states, we present a novel parallel teleportation protocol which teleports multiple quantum states between groups of senders and receivers. 
The notable features of this protocol are that (\emph{i}) the partition into senders and receivers can be chosen after the state has been distributed, and (\emph{ii}) one group has to perform joint quantum operations while the parties of the other group only have to act locally on their system. 
We also prove the equivalence between pure state quantum secret sharing schemes and AME states with an even number of parties.  This equivalence implies the existence of AME states for an arbitrary number of parties based on known results about the existence of quantum secret sharing schemes.
\end{abstract}
\maketitle

{\sl Introduction.}
Entanglement is at the core of the power of quantum information processing and has been extensively studied for few qubits.
The classification of entanglement classes for three and four qubits
is well understood \cite{Dur2000, Verstraete2002, Higuchi2000, Levay2006, Luque2003, Acin2000, Brierley2007} and canonical forms of
pure states under local unitary transformations 
of each local Hilbert space 
have also been analyzed \cite{Acin2000, Kraus2010, Kraus2010a}. As the
number of local quantum degrees of freedom increases, our understanding
of entanglement gets poorer. The number of independent invariants that
classify entanglement grows exponentially and it is unclear which 
purpose each category of entanglement serves \cite{Miyake2002,Gelfand1994}. 
In recent years, there has been 
an important progress in the classification of the maximally multipartite entangled states composed of qubits 
\cite{Osterloh2006,Brown2005,Facchi2008,Gour2010,Brierley2007}.
Nevertheless, a complete understanding of the structure, classification and
usefulness of quantum states with the largest possible entanglement for arbitrary dimension is still missing.
Another motivation for studying multipartite entanglement is its connection to other apparently
unrelated areas of physics, like string theory and black-holes \cite{Borsten2010, Borsten2011}.

Quantum teleportation is one of the most intriguing utilizations of entanglement. It allows distant parties, who share a resource of entanglement, to transport a quantum state from one party to the other by only exchanging classical information and using up said entanglement. The first proposal of such a protocol used the resource of bipartite entanglement between two parties \cite{Bennett1993}. Later teleportation protocols using genuine multipartite entanglement between more than two parties were proposed theoretically for four qubit entanglement \cite{Yeo2006}, and experimentally in the form of open-destination teleportation for five qubits \cite{Zhao2004}.

This manuscript is devoted to initiate the study of a class of states with genuine multipartite entanglement. These states, which we call absolutely maximally entangled (AME) states, are defined
as having the strict maximal entanglement in all bipartitions of the system. 
Up until now, AME states have been thought to be a rather limited concept, because only few AME states exist for qubits \cite{acknowledge}, specifically no AME states exist for four, or eight and more qubits \cite{Gour2010, Rains1999}.
In this work, we consider the \emph{qudit} problem, and show that AME states exist for any number of parties 
by choosing an appropriate qudit dimension. 

The fact that AME states contain
maximal entanglement makes them the natural candidates to 
implement novel multipartite communication protocols. 
Indeed, we shall here show how they can be used to implement novel parallel teleportation scenarios that postpone the choice of senders and
receivers until after the state has been distributed. These protocols require that either the senders or receivers perform joint quantum operations, while the respective other parties only have to act locally on their systems. We further establish a one-to-one correspondence between pure state quantum secret sharing (QSS) schemes \cite{Cleve1999, Gottesman2000} and even-party AME states, which also proves the existence of AME states for \emph{any} number of parties given an appropriate choice of the system dimensions. This follows from the existence of pure state QSS schemes for any odd number of parties \cite{Cleve1999}.
It should be mentioned that, while our parallel teleportation protocol is different from the aforementioned open-destination teleportation, it is also possible to implement open-destination teleportation with AME states \cite{Helwigtobepublished}.

{\sl Definition of AME states.}
An \AME{n,d} state (absolutely maximally entangled state) of $n$ qudits of dimension $d$, $|\psi\rangle\in {\mathbb{C}}_d^{\otimes n}$, 
is a pure state for which every bipartition of the system into the sets $B$ and $A$, with 
$m=|B|\leq |A| = n-m$, is strictly maximally entangled such that
\be
\label{eq:MES-definition}
S(\rho_B) = m \log_2 d \, . 
\ee
Consequently, every partition of $m$ local degrees of freedom
is represented by a reduced density matrix proportional to the identity
\begin{equation}
  \label{definitionAME}
  \rho_B=Tr_{A}|\psi\rangle\langle\psi|=\frac{1}{d^m}I_{d^m} \, , \qquad 1\le m\le \frac{n}{2} .
  \end{equation}
In practice, to detect an AME state it is sufficient to check that 
 all the $\binom{n}{\floor{n/2}}$ possible bipartitions 
of $\floor{n/2}$ qudits are maximally entangled, since all subsequent
traces of the identity matrix are again identity matrices.

A state is an AME state iff it can be written as
\begin{equation}
	\label{eq:defAMEstate}
	\ket{\textrm{AME}} = 
	\frac{1}{\sqrt{d^m}}\sum_{k\in \ZZ_d^{m}} 
	\ket{k_1}_{B_1}\cdots \ket{k_{m}}_{B_{m}}
	\ket{\phi(k)}_A,
\end{equation}
with $\braket{\phi(k)|\phi(k')} = \delta_{kk'}$,
for every partition
into $m = |B| \leq |A| =n-m$ disjoint sets $B$ and $A$.

Two obvious examples of AME states are the Einstein-Rosen-Podolsky (EPR) and the Greenberger-Horne-Zeilinger (GHZ) states for two and three qubits, respectively.
In both cases, the entanglement entropy is maximal for all their partitions. 
It has been proven that there are no absolutely maximally entangled states for four qubits \cite{Gour2010}. AME states exist for five and six qubits \cite{Borras2007}, and a possible form for them will be given later in Example~\ref{example:QSStoAME}. No AME states exist for eight or more qubits \cite{Gour2010, Rains1999}. The existence of an \AME{7,2} state is still an open question, but it has been conjectured in Ref~\cite{Borras2007} that no such state exists. By increasing the party dimension, AME states can be found for these cases in which no qubit AME states exist. We remark, however, that, although we will show that for each $n$, \AME{n,d} states exist for some appropriate choice of $d$, finding the conditions for the
existence of AME($n$,$d$) states, depending on $n$ \emph{and} $d$, is generally a
non-trivial problem. In a future publication \cite{Helwigtobepublished}, we will show that, interestingly, a special class of AME states can be constructed from certain classical error correcting codes, namely those that satisfy the singleton bound \cite{MacWilliams1977}.

{\sl Parallel Teleportation.}
The maximal entanglement property of an \AME{n,d} state for any bipartition into the sets $A$ and $B$ can be used to teleport quantum states between those two sets. In contrast to the teleportation scenario where $A$ and $B$ share a maximally entangled state that is not an AME state, in the AME scenario the sets $A$ and $B$ do not have to be specified when the state is created, but instead can be chosen after the AME state has been distributed.

There are essentially three different ways in which the teleportation can be performed, depending on which parties can perform joint quantum operations, and which are separated and only able to perform local operations on their own quantum systems. 

In the first case, the parties within each set, $A$ and $B$, are able to perform joint quantum operations. A standard teleportation of an arbitrary $d^m$-dimensional state, where $m=\min(|A|,|B|)$, can be performed in either direction.

In the second case, the sending parties $A$ can perform a joint quantum operation, but the parties in $B$ are only able to perform local quantum operations. Additionally we require $m=|B|\leq|A|=n-m$. Then one qudit can be teleported from $A$ to each of the parties in $B$, and thus in total $m$ qudits are teleported from $A$ to $B$. This is illustrated in the left hand side of Figure~\ref{figure:teleport}.

In the third and probably the most interesting case, the sending parties can only perform local operations, but the receiving parties can perform joint quantum operations. In this case, a teleportation is possible if the number of receiving parties is larger or equal $n/2$. Hence, sticking to our convention $m=|B|\leq|A|$, we now consider a teleportation from $B$ to $A$. See the right hand side of Figure~\ref{figure:teleport} for an illustration.

\begin{figure}[t]
	\centering
	\includegraphics[scale=0.53]{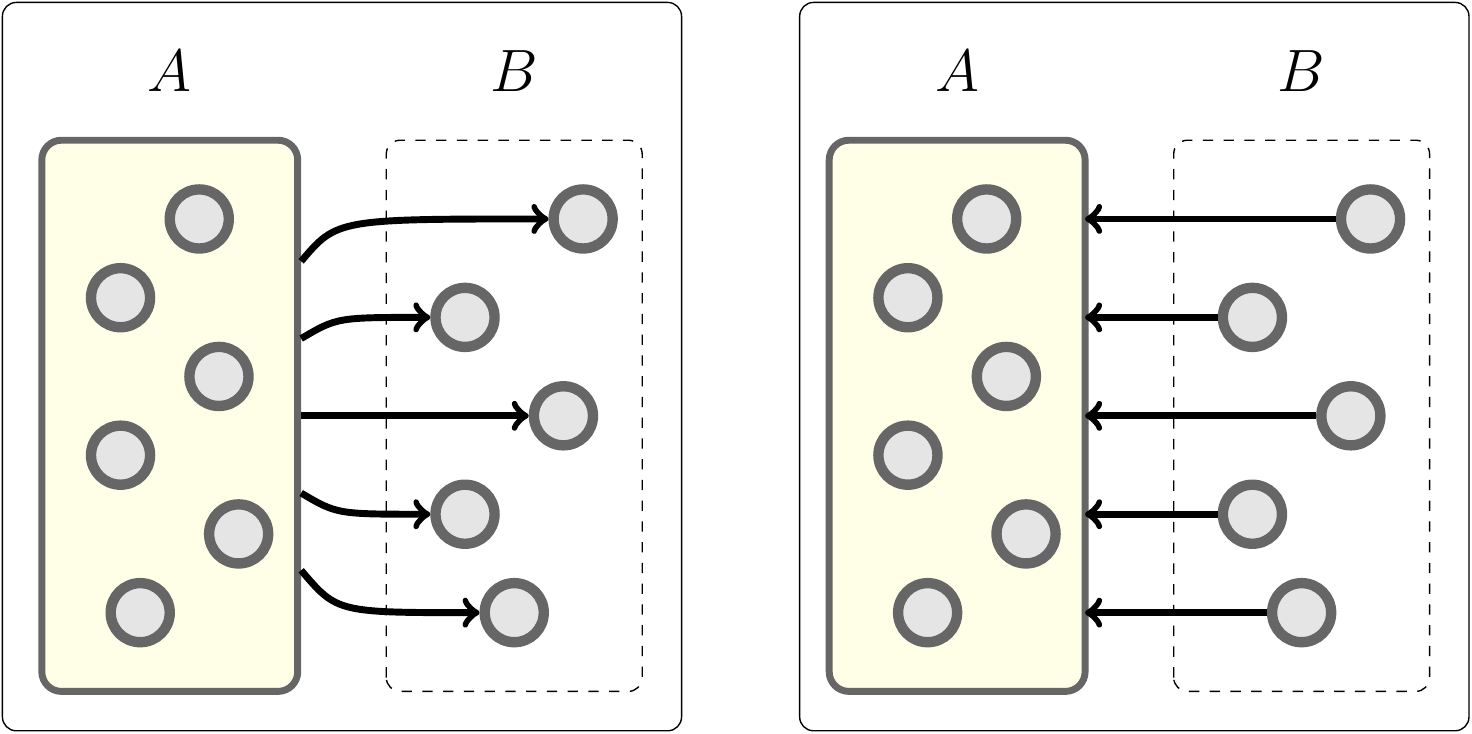}
	\caption{Parallel Teleportation scenarios of Theorem~\ref{theorem:tele}. Scenario (i) is on the left, and (ii) on the right. Parties in $A$ perform joint quantum operations, parties in $B$ only local quantum operations.}
	\label{figure:teleport}
\end{figure}

The first scenario is just a straightforward teleportation between maximally entangled parties. The second and third scenarios are presented in the following theorem.

\begin{theorem}
	\label{theorem:tele}
	Given an \AME{n,d} state, and a bipartition of the $n$ parties into the sets $A$ and $B$ such that $m = |B| \leq |A| = n-m$, then the following two parallel teleportation scenarios are possible
	\begin{enumerate}[(i)]
	 \item $A$  can teleport one qudit to each party in $B$ by performing a joint quantum operation and communicating two classical ``dits" to each party in $B$. Each party in $B$ can then locally recover their respective qudit with a local operation.
	 \item Each party in $B$ can locally teleport one qudit to $A$. After receiving the measurement outcomes, consisting of two ``dits'' of classical information from each party in $B$, the parties in $A$ are able to recover all $m$ qudits by performing a joint quantum operation.
	\end{enumerate}
\end{theorem}

\begin{proof}
In both scenarios the parties in set $A$ perform a joint quantum operation to transform the AME state into $m$ $d$-dimensional EPR pairs. Then these pairs are used to teleport $m$ qudits from the sending to the receiving parties. This is done by performing the joint unitary operation 
\begin{equation}
	\label{eq:tele-Ua}
	U_A \ket{\phi(k)}_A = \ket{k_1}_{A_1} \cdots \ket{k_{m}}_{A_{m}}
	\ket{0}_{A'}.
\end{equation}
on the initial \AME{n,d} state
\begin{equation}
	\label{eq:tele-ame}
	\ket{\Phi} = \frac{1}{\sqrt{d^{m}}}
	\sum_{k\in \ZZ_d^{m}} \ket{k_1}_{B_1} \cdots \ket{k_{m}}_{B_{m}}
	\ket{\phi(k)}_A,
\end{equation}
with $\braket{\phi(k)|\phi(k')}=\delta_{kk'}$. This results in the state
\begin{align}
	\label{eq:tele-eprs}
	U_A \ket{\Phi} 
	&= \ket{\Psi}_{B_1 A_1} \cdots \ket{\Psi}_{B_m A_m}
	\ket{0}_{A'},
\end{align}
where $\ket{\Psi} = \sum \ket{i}\ket{i}$ are $d$-dimensional EPR pairs. These EPR pairs can now be used to teleport a qudit from $A_i$ to $B_i$ in case (i) ($B_i$ to $A_i$ in case (ii)). This requires $A_i$ ($B_i$) to perform a generalized Bell measurement on her qudit and the qudit she wants to teleport, and communicate the measurement result to $B_i$ ($A_i$). This amounts to sending the classical information of two ``dits" for each EPR pair. Upon reception of the measurement result, $B_i$ ($A_i$) can recover the teleported qudit by performing an appropriate unitary on his qudit.
\end{proof}

{\sl Quantum Secret Sharing.}
The last teleportation scenario suggests a close relationship between AME states and quantum secret sharing (QSS) schemes \cite{Cleve1999}. 
In a QSS protocol \cite{Cleve1999, Gottesman2000}, a dealer encodes a secret $S$ in a quantum state that is shared among $n$ players in such a way that only special subsets of players are able to recover the secret. The set of all subsets that are able to recover the secret form the access structure and the set of all subsets that can gain no information about the secret form the adversary structure. If the encoded state is a pure state, we call it a pure state QSS scheme. We are only interested in pure state QSS schemes here.

Additionally, we restrict our attention to threshold QSS schemes \cite{Cleve1999}, which means that the access structure is formed by all sets that contain $k$ or more number of parties, while any set with less than $k$ parties cannot obtain any information about the secret. Thus $k$ is the threshold number of parties required to recover the secret. Such a QSS scheme is denoted as a \QSS{k,n}. For pure state threshold QSS schemes, $n$ and $k$ are always related by $n=2k-1$.

To see the relation between AME states and threshold QSS schemes, we consider an \AME{2m,d} state with an even number of parties and divide the parties into two sets $A=\{A_1,\ldots,A_m\}$ and $B=\{D,B_1,\ldots,B_{m-1}\}$ of equal size $m$. In set $B$ we have singled out one party $D$, which will act as the dealer of the QSS scheme. Now we perform the protocol of Theorem~\ref{theorem:tele} (ii), but only $D\in B$ performs the final teleportation operation. Also note that the unitary operation in Equation~\eqref{eq:tele-Ua} and the Bell measurement by the dealer commute. Thus, $D$ can first perform her Bell measurement, effectively encoding the teleported qudit onto the residual AME state, from which it can be recovered by the players in $A$.

Furthermore, instead of the bipartition into the sets $A$ and $B$, we could have equally well chosen any other bipartition into two sets $A'$ and $B'$ of cardinality $m$ with $D \in B'$. Then, without changing the operations that $D$ has to perform, the parties in $A'$ are able to recover the teleported qudit (see Figure~\ref{fig:AMEtoQSS} for an illustration).

\begin{figure}[t]
	\centering
	\includegraphics[width=\linewidth]{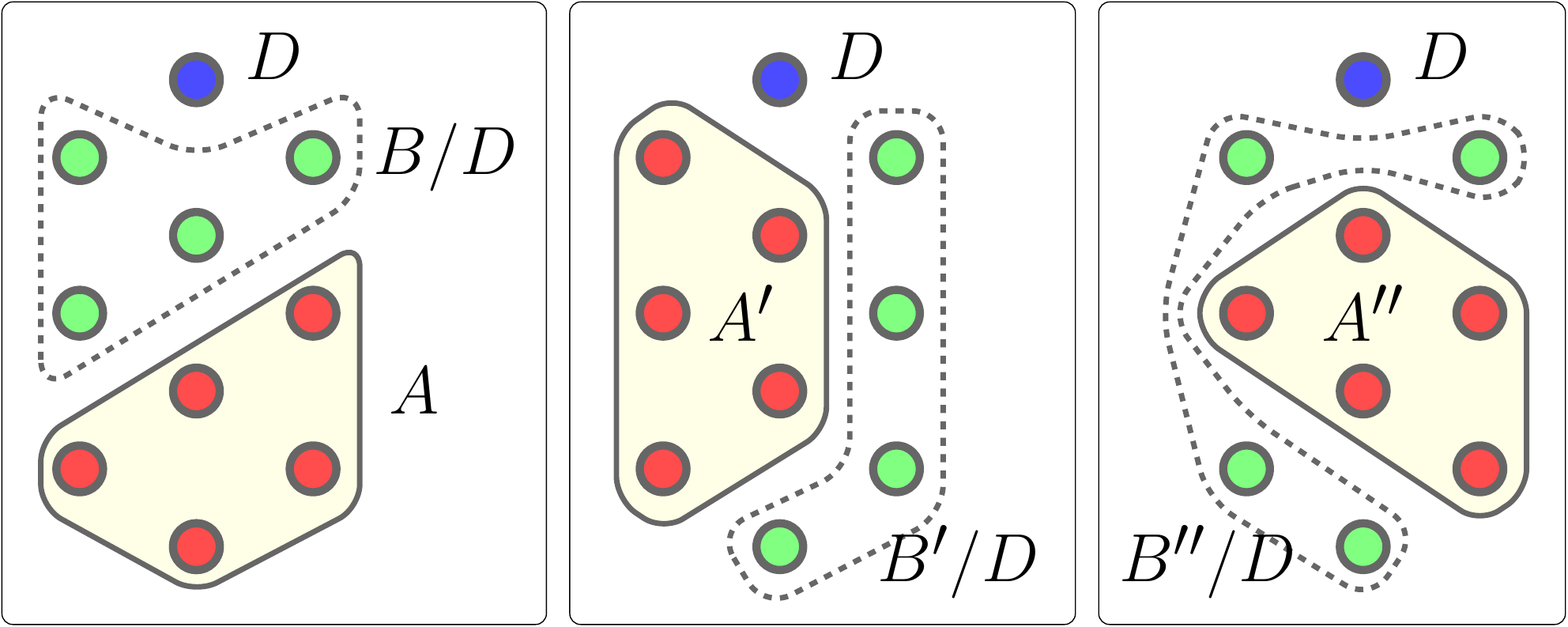}
	\caption{(Color online) After $D$ (blue) performs her teleportation operation, any set of $m$ parties (red), $A$, $A'$, $A''$ etc., can recover the teleported state. Any set of parties with $m-1$ or less parties (any set consisting only of green parties) cannot gain any information about the teleported state.}
	\label{fig:AMEtoQSS}
\end{figure}

Thus, any set of at least $m$ of the residual $2m-1$ parties without $D$ can recover the teleported state, given that the measurement outcome is broadcasted to all parties. Furthermore, the no-cloning theorem guarantees that any set of less than $m$ players cannot gain any information about the state \cite{Gottesman2000}. Hence we accomplished to construct a \QSS{m,2m-1} from an \AME{2m,d} state. 

Before stating the theorem that formulates this observation concisely, we shortly review how a QSS protocol works. A secret of dimension $d$, $\ket{S} = \sum a_i \ket{i}$, is encoded into the state $\sum a_i \ket{\Phi_i}$ which is shared by the players such that each authorized set can deterministically recover $\ket{S}$ from its reduced state, while the reduced state of unauthorized sets is independent of the encoded secret. We call $\ket{\Phi_i}$ the basis  states of the QSS scheme, and we show in  \cite{Helwigtobepublished} that they are AME states for pure state threshold QSS schemes with equal share and dimension size.

\begin{theorem}
	\label{theorem:AME-QSS}
There is a one to one correspondence between an \AME{2m,d} state and a pure state $((m,2m-1))$ threshold QSS scheme, whose share and secret dimensions are $d$.
\end{theorem}

\begin{proof}
\emph{AME to QSS}: 
For any bipartition into parties $A=\{A_1,\ldots,A_m\}$ and $B=\{D,B_1,\ldots,B_{m-1}\}$, the \AME{2m,d} states has the form
\begin{equation}
	\ket\Phi = 
	\frac{1}{\sqrt{d^m}}\sum_{(i,k)\in \ZZ_d^{m}} 
	\ket{i}_D \ket{k_1}_{B_1}\cdots \ket{k_{m-1}}_{B_{m-1}}
	\ket{\phi(i,k)}_A \, ,
   \nonumber
\end{equation}
with $\braket{\phi(k,i)|\phi(k',j)}=\delta_{kk'}\delta_{ij}$.
We define the QSS basis states 
\begin{align}
	\label{eq:AMEtoQSS-qss}
	\ket{\Phi_i} &= \sqrt{d}\ {}_{D}\!\!\braket{i|\Phi} \nonumber \\
	&= \frac{1}{\sqrt{d^{m-1}}}
	\sum_{k\in \ZZ_d^{m-1}} \ket{k_1 \cdots k_{m-1}}_{B}
	\ket{\phi(k,i)}_A.
\end{align}
A secret encoded as 
\begin{equation}
	\ket{a} = \sum a_i \ket{i} \rightarrow
	\sum a_i \ket{\Phi_i},
\end{equation}
satisfies the requirement of a threshold QSS scheme, because the parties $B$ have a completely mixed states, independent of the encoded secret. Additionally, the set $A$, which can be chosen to be any set of $n$ players, can restore the secret $\ket{a}$ by performing the joint unitary operation
\begin{equation}
	\label{eq:AMEtoQSS:Usecret}
	U_A \ket{\phi(k,i)}_A = \ket{k_1}_{A_1} \cdots \ket{k_{m-1}}_{A_{m-1}} \ket{i}_{A_m}.
\end{equation}

\emph{QSS to AME}:
For any bipartition into $m$ authorized parties $A=\{A_1,\ldots,A_m\}$ and $m-1$ unauthorized parties $B=\{B_1,\ldots,B_{m-1}\}$, the AME basis states of the QSS scheme can be written in the form
\begin{equation}
	\ket{\Phi_i} = \frac{1}{\sqrt{d^{m-1}}}
	\sum_{k\in \ZZ_d^{m-1}} \ket{k_1}_{B_1} \cdots \ket{k_{m-1}}_{B_{m-1}}
	\ket{\phi(k,i)}_A,
	\nonumber
\end{equation}
where $\braket{\phi(k,i)|\phi(k',i)}=\delta_{kk'}$, because the states are AME states, and $\braket{\phi(k,i)|\phi(k,j)}=\delta_{ij}$, because the authorized parties can recover the secret deterministically. Thus,
\begin{equation}
	\label{eq:QSStoAME-phi_ortho}
	\braket{\phi(k,i)|\phi(k',j)}=\delta_{kk'}\delta_{ij}.
\end{equation}
From these basis states, define the state
\begin{align}
	\ket\Phi &= \frac{1}{\sqrt{d}}\sum_{i\in\ZZ_d} \ket{i}\ket{\Phi_i}
	\nonumber \\
	&= 
	\frac{1}{\sqrt{d^m}}\sum_{(i,k)\in \ZZ_d^{m}} \ket{i}_D \ket{k_1}_{B_1}\cdots \ket{k_{m-1}}_{B_{m-1}}
	\ket{\phi(k,i)}.
	\nonumber
\end{align}
Because of Equation~\eqref{eq:QSStoAME-phi_ortho}, $\ket\Phi$ is a maximally entangled state with respect to the bipartition $B\cup \{D\}$ vs. $A$. Since the original bipartition into $A$ and $B$ was arbitrary, $\ket\Phi$ is maximally entangled with respect to any bipartition into two cardinality $m$ sets and thus is an \AME{2m,d} state.
\end{proof}

Since it is known that \QSS{m,2m-1} exist for any number of $m$ and an appropriate choice of $d$ \cite{Cleve1999}, 
Theorem~\ref{theorem:AME-QSS} proves the existence of AME states for any number of parties.

\begin{example}
\label{example:QSStoAME}
In this example, we show how the five qubit code can be used to construct \AME{5,2} and \AME{6,2} states. From the five qubit code a \QSS{3,5} can be constructed \cite{Cleve1999}. The corresponding basis states are 
\begin{gather}
\begin{split}
 \ket{0_L}= \frac{1}{4}
 (& \ket{00000}+\ket{10010}+\ket{01001}+\ket{10100}\\
 +& \ket{01010}-\ket{11011}-\ket{00110}-\ket{11000}\\
 -& \ket{11101}-\ket{00011}-\ket{11110}-\ket{01111}\\
 -& \ket{10001}-\ket{01100}-\ket{10111}+\ket{00101}),
\end{split}\\
\begin{split}
 \ket{1_L}= \frac{1}{4}
 (& \ket{11111}+\ket{01101}+\ket{10110}+\ket{01011}\\
 +& \ket{10101}-\ket{00100}-\ket{11001}-\ket{00111}\\
 -& \ket{00010}-\ket{11100}-\ket{00001}-\ket{10000}\\
 -& \ket{01110}-\ket{10011}-\ket{01000}+\ket{11010}).
\end{split}
\end{gather}
These states are \AME{5,2} states as required. Following the receipe of Theorem~\ref{theorem:AME-QSS}, we obtain the \AME{6,2} state
\begin{gather}
\begin{split}
 \ket{\Phi}= & \frac{1}{\sqrt{2}}[\ket{0}\ket{0_L}+\ket{1}\ket{1_L}]\\
=&\frac{1}{4}[\ket{000}(\ket{+-+}+\ket{-+-})\\
&+\ket{001}(-\ket{+--}+\ket{-++})\\
&+\ket{010}(\ket{++-}-\ket{--+})\\
&+\ket{011}(-\ket{+++}-\ket{---})\\
&+\ket{100}(-\ket{+++}+\ket{---})\\
&+\ket{101}(-\ket{++-}-\ket{--+})\\
&+\ket{110}(-\ket{+--}-\ket{-++})\\
&+\ket{111}(-\ket{+-+}+\ket{-+-})].
\end{split}
\end{gather}
\end{example}

{\sl Conclusion.}
In this manuscript, we have introduced AME states, a class of highly entangled states, for $n$ qudits shared among $n$ locally separated parties. 
Previous investigations of maximal entanglement showed that AME states do not exist when the number of qubits is eight or larger. Here we proved the existence of AME states for any number of parties with the appropriate qudit dimension. Moreover, we have shown how they can be utilized in different parallel teleportation scenarios, which require some parties to perform joint quantum operations, while others' capabilities may be restricted to local operations. In those scenarios the advantage of AME states over less entangled states like a collection of EPR pairs lies in the fact that the partition into senders and receivers may be chosen after the state has been distributed. 

Furthermore, we have investigated the relationship of AME states with QSS schemes and established a one-to-one correspondence between even party AME states and pure state threshold QSS schemes. 
This correspondence allows us to prove the existence of AME states for any number of parties with the appropriate dimension.
In future work we further explore this very intuitive approach to develop new communication protocols from AME states as well as extending the range of QSS schemes that can be derived from AME states. For instance, instead of assigning the role of the dealer to only one of the parties in the AME state, we can imagine multiple dealers who encode independent secrets onto the residual AME states, resulting in QSS schemes with more involved access structures.
The established connection to QSS schemes also confirms a relation between AME states and quantum error correction codes that was
already suggested in Ref.~\cite{Scott2004}.
A better understanding of this relation will allow us to construct new quantum error correction codes from AME states as well as deriving AME states from already known quantum codes.
This might also shed light upon the open question of existence of AME states for a given number of parties and system dimension.

{\sl Acknowledgments.}
W.H., W.C., and H.K.L. acknowledge financial support from funding agencies including NSERC, Quantum-Works, the CRC program and CIFAR. J.I.L. and A.R. acknowledge financial support from  MICIN (Spain) and Grup consolidat (Generalitat de Catalunya).
We would also like to thank David Gosset and Sandu Popescu for very helpful comments.

\bibliography{AME_QSS}

\end{document}